\documentclass{llncs}
\usepackage[utf8]{inputenc}
\usepackage{amsmath,amssymb,mathrsfs,amstext,amsfonts}
\usepackage{extarrows}
\usepackage{braket}
\usepackage{xcolor}
\usepackage{soul}
\usepackage{IEEEtrantools}
\usepackage{graphicx}
\usepackage{tikz}
\usepackage{tabularx}
\usepackage{stmaryrd}
\usepackage{comment}
\usepackage{dsfont}

\usepackage{mdframed}
\usepackage{dirtytalk}
\usepackage{algorithm}
\usepackage{algorithmic}
\usepackage{graphicx}
\usepackage{float}
\usepackage{authblk}
\usepackage{enumitem}
\usepackage{multicol}
\usepackage{url}

\newcommand{\A}{\mathcal{A}}
\newcommand{\K}{\mathcal{K}}

\newcommand{\QL}{\mathcal{QL}}

\newcommand{\R}{\mathcal{R}}
\renewcommand{\S}{\mathcal{S}}

\newcommand{\D}{\mathcal{D}}
\newcommand{\Q}{\mathcal{Q}}

\newcommand{\splitatcommas}[1]{%
    \begingroup
    \begingroup\lccode`~=`, \lowercase{\endgroup
        \edef~{\mathchar\the\mathcode`, \penalty0
            \noexpand\hspace{0pt plus 1em}}%
        }\mathcode`,="8000 #1%
    \endgroup
}

\newtheorem{assumption}{Assumption}

\newcounter{protocol}

\begin{document}
\title{Establishing shared secret keys on quantum line networks: protocol and security}

\author{Mina Doosti\inst{1} \and Lucas Hanouz\inst{2} \and Anne Marin\inst{2} \and Elham Kashefi\inst{1,3} \and Marc Kaplan\inst{2}\footnote{The first two authors contributed equally to the work. The last author supervised the project.}}
\institute{
School of Informatics, University of Edinburgh, 10 Crichton Street, Edinburgh EH8 9AB, United Kingdom
\and
VeriQloud, 9bis rue Abel Hovelacque, 75013 Paris, France
\and
Laboratoire d’Informatique de Paris 6, CNRS, Sorbonne Universit\'e, 4 Place Jussieu, 75005 Paris, France
}
\maketitle

\begin{abstract}
We show the security of
multi-user key establishment on a single line of quantum communication.
More precisely, we consider a quantum communication architecture where the qubit generation and measurement happen at the two ends of the line,
whilst intermediate parties are limited to single-qubit unitary transforms.
This network topology has been previously introduced
to implement quantum-assisted secret-sharing protocols for classical data, as well as the key establishment, and secure computing.

This architecture has numerous advantages. The intermediate nodes are only using simplified hardware, which makes them easier to implement. Moreover, key establishment 
between arbitrary pairs of parties in the network does not require key routing through intermediate nodes.
This is in contrast with quantum key distribution (QKD) networks for which non-adjacent nodes
need intermediate ones to route keys, thereby revealing these keys to intermediate parties and consuming previously established
ones to secure the routing process.

Our main result is to show the security of key establishment on quantum line networks.
We show the security using the framework of abstract cryptography.
This immediately makes the security composable, showing that the keys can be used
for encryption or other tasks.

\end{abstract}

\section{Introduction}

The first discovered application of quantum communication was Quantum Key Distribution (QKD),
published in 1984 by Bennett and Brassard~\cite{BB84}.
It shows that quantum resources can achieve unconditionally secure key establishment, a task that is provably impossible using only classical resources.
With the BB84 protocol, no matter what effort eavesdroppers spend, they will not learn a single bit of the secret key established by the protocol.

The possibility of unconditional security with quantum resources is a neat mathematical statement. However, building a real-world security infrastructure using quantum key distribution requires considering many details in addition to a quantum communication channel.
Firstly, it requires authenticated communication channels, which might need, in order to achieve
unconditional security, a short shared secret key in the first place. Secondly, the specifications of the original protocols are hardly met with
real-world's hardware. Last but not least, the advantage of QKD is harder to characterize when considering topologies beyond point-to-point communications.

Despite these fair criticisms, the scientific importance of quantum key distribution led to its unfortunate
identification with the wider area of quantum communication.
While there are lots of challenges to improving the performances of QKD protocols, this approach to quantum communication has narrowed down the concept of quantum networks to a collection of nodes running key establishment protocols on each edge.

This quantum communication architecture bears a paradox which currently strongly limits the use-cases for quantum communication.
This paradox consists of two contradictory features of quantum networks. On the one hand, fibre losses strongly limit the achievable distance of quantum communication.
On the other hand, since each edge of the network requires dedicated quantum hardware, the scaling of its cost is super-linear in its size. This tends to limit the deployment of quantum hardware to ``critical'' sites which may be far from each other. 

The two approaches that are usually pursued to solve this paradox are to either extend the range of quantum communication or to decrease QKD's cost by improving the underlying hardware. 
Regarding the distance, the standard solution is to use so-called trusted nodes to increase the range. These intermediate nodes are in charge of routing encryption keys to distant parties and consequently, get a clear view of the information they are routing.
Therefore, the trust in trusted nodes does not refer to cryptography but rather stands on physical and social grounds i.e. physical protection from a potential eavesdropper.
Other approaches to increase the range include the use of satellites, but their cost also limits them to critical infrastructure. A completely different option lies in the development of quantum repeaters, which will lead to a change of paradigm for quantum networks when available. They are, however, much further away in terms of technology readiness.

In this work, we consider a different approach to scaling quantum networks.
We consider quantum line networks, called \emph{Qlines} in the rest of this paper.
A Qline consists in an initial node that can generate single qubit states and a final node that can measure them. In between, intermediate nodes only have the ability to perform single-qubit transforms. 
This architecture has the following features:
\begin{itemize}
\item The intermediate nodes only use very limited hardware and can be implemented with standard telecom components (i.e. optical modulators).
\item Any pair of nodes on the network can establish a shared secret key.
Unlike  QKD networks, this key establishment does not require routing between
non-adjacent nodes, avoiding consuming keys for the sake of routing.
\item Since there is no need for key routing, intermediate parties never learn the keys of other parties. 
Compared to the QKD networks, this drastically lowers the trust assumptions and enhances the security of the network.
\end{itemize}

Overall, Qline introduces a new tradeoff between security, performance and cost, for key establishment. For this reason, it can potentially address use-cases that are currently out of range of the standard QKD network approach.

Architectures similar to Qline architecture has initially been introduced for quantum-assisted secret sharing~\cite{HBB99}. This task consists, given some secret data $d$ and $n$~parties, in the establishment of $n$ secret shares $s_1, \ldots, s_n$ distributed among the $n$ parties, and such that $d$ can be reconstructed from any subset of $k$ shares, but not from subsets of $k-1$ shares. The integer $k$~is called the threshold, and with the parameters used above, the scheme is usually referred to as an $(n,k)$-secret sharing scheme.

The classical version of secret sharing has been discovered independently by Shamir and Blakley in 1979~\cite{S79,B79}. Several years later, Hillery, Buzek and Berthiaume
showed that it was possible to have an $(n,n)$ secret sharing scheme by distributing GHZ entangled states between the parties involved~\cite{HBB99}. 
Later, it was shown that the same result could be obtained using only single-qubit states, with participants on a line successively transforming the qubits~\cite{STWB04}. Formally, the network topology that it considers is similar to a Qline.

Quantum-assisted secret sharing\footnote{In the literature, the task is usually referred to as \emph{Quantum secret sharing}. In our view, the name should reflect the fact that it is for classical secrets, and assisted by quantum communication.} was then extended in several different directions, such as sharing quantum secrets~\cite{CDL99}, using continuous-variable states, or multi-level states~\cite{XG13}. More importantly for us, the underlying architecture has been shown to allow other protocols. In particular, it allows any pair of parties to establish a shared secret key, like in standard QKD~\cite{GELL15}.
In this work, we use a similar strategy for key establishment, which is to first perform a secret sharing protocol and then reveal some information.

Beyond QKD, Qline has also been used for quantum-assisted secure multiparty computing~\cite{CPEW17}. In this task, several parties compute a function without revealing information on their inputs. Secret sharing is an example of such a task, but it is interesting to note that other tasks have been discovered and implemented. This seems to indicate that there are still more Qline protocols to be discovered.

The main contribution of our work is to provide a complete security proof of the key establishment protocol over Qline.
Previously, a security proof has been provided for quantum-assisted secret sharing~\cite{KXHA17}. The context, however, is rather different. In multiparty computing in general, each party wants to protect against malicious players that are trying to gain access to more information than needed to solve the task. In QKD, the security guarantee is only against eavesdropping.
This allows us to rely on standard tools
for QKD security, in particular the \emph{abstract cryptography} framework~\cite{AbstractCrypto}.

While the security of key exchange over Qline may seem obvious, providing formal proof has numerous advantages. Firstly, the main argument to use QKD is that, unlike any classical key establishment protocol, it provides provable security. Key establishment over Qline provides the same level of security but claiming it without a formal proof would be paradoxical. Secondly, the rich history of QKD security proofs has shown that resting on intuition may be misleading~\cite{KRBM07}. Last but not least, using the abstract cryptography framework immediately makes our proof composable, 
allowing to combine the keys established over Qline with authentication, encryption or other key encapsulation mechanisms~\cite{DHP20}.

The content of the paper is the following. 
In Section~\ref{sec:prelim}, we introduce the notations for quantum information processing and the abstract cryptography framework.
In section~\ref{sec:protocol}, we describe the architecture of Qline and introduce the quantum-assisted secret sharing and key establishment protocol. In Section~\ref{sec:proof}, we formally prove the security of the protocol.

\section{Preliminaries}
\label{sec:prelim}
\subsection{Quantum information processing for key establishment}
In this section, we present the basics of quantum information processing required to understand our results. Further information can be found in standard textbooks~\cite{NC02}.

We consider quantum systems whose description is a vector in a Hilbert space $\mathcal H$ of dimension $d\in \mathbb N$. We use the \emph{ket} notation $\ket \phi$ to denote those states. For $d=2$, the quantum systems are called \emph{qubits}. The following states are usually called the BB84 states:

\[
\ket 0 = \left ( \begin{array}{c} 0 \\1 \end{array}\right ),
\ket 1= \left(\begin{array}{c} 1 \\0 \end{array}\right),
\ket + =\frac 1 {\sqrt 2} \left(\begin{array}{c} 1 \\1 \end{array}\right),
\ket - =\frac 1 {\sqrt 2} \left(\begin{array}{c} 1 \\-1 \end{array}\right).
\]

Systems composed of multiple qubits are described using tensor products:
$\bigotimes_{i=1}^N \ket {x_i}=\ket {x_1}\otimes...\otimes \ket {x_N} \in \mathcal H_1 \otimes \ldots \otimes \mathcal H_N$. For brevity, we usually omit the tensor product sign $\otimes$ when using multiple qubit systems.

State transformations are described by unitary matrices. In particular, the computational basis formed by $\{\ket0, \ket 1\}$ is mapped to the Hadamard basis $\{\ket +, \ket - \}$ through the Hadamard transform $H$.

Our key establishment protocol uses \emph{rotations around the Y axis}. More precisely, we will use
\[
R_Y(\theta) = e^{i\theta/2} \begin{bmatrix} \cos(\theta/2) & -\sin(\theta/2) \\ \sin(\theta/2) & \cos(\theta/2) \end{bmatrix}\]
In particular, we consider the following angles:
\begin{flalign*}
R_Y(\pi) &= Y = i \begin{bmatrix} 0 & -1 \\ 1 & 0 \end{bmatrix},\\
R_Y(\pi/2) &= \sqrt{Y} = \frac{1+i}{2} \begin{bmatrix} 1 & -1 \\ 1 & 1 \end{bmatrix} = e^{i\pi/4} \frac{I - iY}{\sqrt{2}},\\
R_Y(3\pi/2) &= Y^{3/2} = \frac{1-i}{2} \begin{bmatrix} 1 & 1 \\ -1 & 1 \end{bmatrix} = e^{-i\pi/4} \frac{I + iY}{\sqrt{2}}.
\end{flalign*}

We first introduce projective measurements in the computational basis. Given a qubit $\ket \phi$, a measurement in the computational basis $\{\ket 0, \ket 1\}$ returns 0 or~1 with probability $\vert \braket{0\vert\phi} \vert ^2$ and $\vert \braket {1\vert \phi} \vert ^2$, respectively. A measurement in the Hadamard basis is performed by first applying $H$ to the state and then measuring in the computational basis.

Using these concepts, the BB84 protocol can be described in simple terms. Alice sends one of the BB84 states to Bob, who chooses a measurement basis at random between computational and Hadamard. If Alice and Bob have chosen the same basis, then the measurement's result is known only to them; otherwise, Bob has a purely random bit.

After repeating the operation a number of times, they can post-select the bits corresponding to a similar choice of basis. This phase is called \emph{key sifting}. Usually, it is followed by a post-processing phase that accounts for errors and potential eavesdropping. We don't describe those phases in full detail here since it is not required to prove our results.

We use the following notations for describing the random choices of Alice and Bob. Alice's choice of basis is encoded into a sequence of bits $(b^A_j)$. The
sequence $(b^A_i)_{i \in I}$ is the sub-sequence of $(b^A_i)$ restricted to indices $i$ such that $i\in I.$

While these notations are sufficient to fully describe the protocol, we need more general concepts to tackle the proof. More specifically, we describe the quantum systems with their \emph{density matrices} and use \emph{generalized measurements}.

A density matrix is an operator of the form $\rho = \sum_j p_j \ket{\psi_j}\bra{\psi_j}$, with $p_j>0$ and $\sum_j p_j =1$. Intuitively, this corresponds to preparing the state $\ket{\psi_j}$ with probability $p_j$.

A generalized measurement is defined as a set of operators $\{M^a\}_{a}$
such that $M^a \geq 0$ and $\sum_a M^a =Id$. By definition, the measurement of the state $\rho$
leads to output $a$ with probability $Tr[M^a \rho]$.

\subsection{The abstract cryptography framework}

We establish our security proofs using the \emph{Abstract Cryptography} framework. Initially introduced by Ueli Maurer and Renato Renner~\cite{AbstractCrypto}, this framework is designed to prove the composable security of cryptographic constructions while remaining as general as possible concerning security notions. In opposition to the traditional bottom-up approach where one would first define a low-level computational model (i.e Turing machines, classical communications) and then define notions like computational complexity and cryptographic security, the abstract cryptography framework adopts a top-down approach. It first defines at the highest level the composition of abstract systems and then proceeds downwards by introducing in each lower level only the minimal necessary specifications.

In this framework, systems are defined as abstract objects with \textit{interfaces} where all their potential inputs and outputs are defined. The overall behaviour of the system, that is, the function mapping inputs to outputs, 
is also part of its definition. 

Systems can be composed, either in parallel or sequentially.
The parallel composition of two systems $\R$ and $\S$ is a new system denoted $R||S$
with the interfaces of both sub-systems. It simply describes the fact that these systems are put side by side and seen as a whole, unique system. The behaviour of the composed system is naturally defined from the independent behaviours of the sub-systems (c.f Figure~\ref{fig:def-composition}).

\begin{figure}[!ht]
    \centering 
    \includegraphics[scale=0.50]{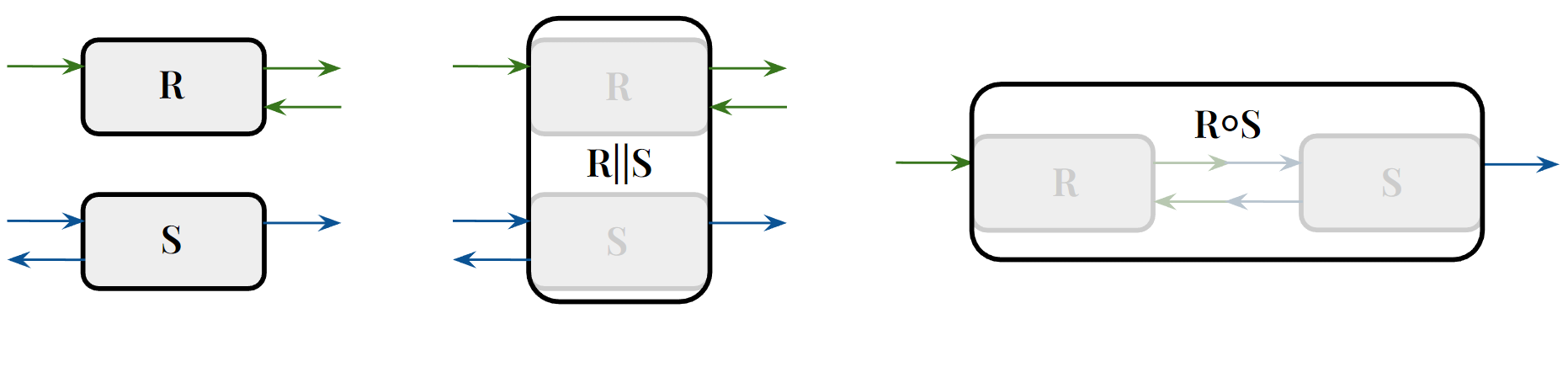}
    \caption{\centering  Composition of abstract systems}
    \label{fig:def-composition}
\end{figure}

The sequential composition describes the fact that the outputs of a system can be used as input by other systems. For instance,
two systems $\R$ and $\S$ can be sequentially composed \textit{at the interfaces $i_R$ of $\R$ and $j_S$ of $\S$} if each input (respectively output) of these interfaces can be associated with a unique output (respectively input) of the other interface. Since in this work, it will always be clear how and at which interfaces a sequential composition occurs, we denote it  $R\circ S$ or simply $RS$ without specifying the interfaces $i_R$ and $j_S$. The resulting system has all the interfaces of both sub-systems except from $i_R$ and $j_S$. The behaviour of the composed system is naturally defined from the behaviour of the sub-systems, knowing that the inputs at the interfaces $i_R$ and $j_S$ are the outputs of the other interface.

In this setting, the security of a cryptographic scheme is  defined as the ``closeness" of that system to its ideal and flawless, imaginary version. This closeness is measured using a pseudo-metric $\delta$ which is an \textit{a priori} unspecified function only assumed to verify the standard properties: identity, symmetry and triangle inequality.

A protocol $P$ is called \textit{$\epsilon$-secure} if there exists a system $\sigma$ called \textit{simulator} such that $\delta(P, \tilde P \circ \sigma) \leq \epsilon$ with $\tilde P$ the ideal version of the protocol. We also denote it $P \approx_\epsilon \tilde P\circ \sigma$. The definition of security for a cryptographic protocol thus depends on two factors: the definition of the ideal version of the protocol, and the pseudo-metric chosen to measure the distance to that ideal version.

In this work, we use the same pseudo-metric as in standard QKD security proofs~\cite{rennerabstract},
namely the \textit{distinguishing} pseudo-metric. The pseudo-metric $\delta$ between two systems is defined as the maximum probability for a computationally unbounded entity called a \textit{distinguisher} to distinguish the systems (over all possible distinguishers). The distinguisher has control over all inputs and access to all outputs of the systems.

This pseudo-metric leads to a composable definition of security, meaning that the composition of secure systems remains secure.
This comes from theorem~1 of~\cite{AbstractCrypto}.
In particular, this means that an $\epsilon$-secure system
can replace its ideal version in any setting with no discernible effect for adversaries with abilities comparable to the distinguisher, except with probability $\epsilon$ (see~\cite{AbstractCrypto}, theorem $2$).

Typically, the cryptographic security of a protocol is proved by
performing the following steps:
 \begin{enumerate}
     \item Model the studied protocol as a system $P$. Each interface contains the inputs and outputs used by the entities operating the protocol. In particular, an \textit{outer} (or \textit{adversarial}) interface describes what external entities can observe and how they can interact with the protocol (e.g for a communication channel, listening and tampering with the communications).
     
     \item Define an ideal system $\tilde P$ which satisfies the desired security properties. $\tilde P$ is not a model of an actual object. It simply defines the desired functionality. The security is often captured by the fact that the \textit{outer} interface contains few or no inputs and outputs for external entities to interact with the system. 
     
     \item Define a \textit{simulator} $\sigma$. It is a system that plugs into the outer interface of the ideal construction and instead provides the same outer interface for external entities as system $P$. 
     The simulator is also required to emulate the behaviour of the outer interface of $P$ such that $\tilde P \circ \sigma$ and $P$ are close enough under the pseudo-metric $\delta$.
     
     \item Finally, prove $P \approx_\epsilon \tilde P \circ \sigma$ for some $\epsilon >0$.
     Using the distinguishing metrics, it amounts to proving that no distinguisher can distinguish $P$ from $\tilde P \circ \sigma$, except with probability~$\epsilon$.
 \end{enumerate}

\section{Quantum communication with a single communication channel}
\label{sec:protocol}



\subsection{Description of the Qline architecture}
In Qline, the nodes of the network are connected on a single quantum communication channel.
The first party is called Alice and the last one is called Bob. The intermediate ones are called the Charlies, usually followed by a number: Charlie 1, Charlie 2, etc.
Each name refers to a specific role in the network.
\begin{itemize}
\item Alice can generate any single-qubit quantum state,
\item Bob can measure single-qubit quantum states in any base,
\item The Charlies can apply arbitrary single-qubit unitaries on incoming qubits.
\end{itemize}
These different capabilities translate into different hardware.
While Alice may need some single-photon source and Bob some single-photon detector, the Charlies only need to modulate the photons that they receive and send. For instance, if qubits are encoded in the polarization of the photons, the Charlies only need polarisation modulators.

\begin{figure}[ht]
\begin{center}
\begin{tikzpicture}
	\fill[gray] 	(-0.5,0) arc (270:90:0.5cm) ;
	\fill[gray]	(-0.6,0) rectangle (0,1);
	\draw	(0,0.5) -- (1.5,0.5) ;
	\fill[gray] 	(1.5,0) rectangle (2.5,1);
	\draw	(2.5,0.5) -- (4,0.5) ;
	\fill[gray] 	(4,0) rectangle (5,1);
	\draw	(5,0.5) -- (6.5,0.5) ;
	\fill[gray]	(6.5,0) rectangle (7.1,1);
	\fill[gray] 	(7,1) arc (90:-90:0.5cm) ;
	\node at (-0.5,-0.3) {Alice};
	\node at (7,-0.3) {Bob};
	\node at (2,-0.3) {Charlie 1};
	\node at (4.5,-0.3) {Charlie 2};
\end{tikzpicture}
\caption{A Qline with four nodes}
\end{center}
\end{figure}
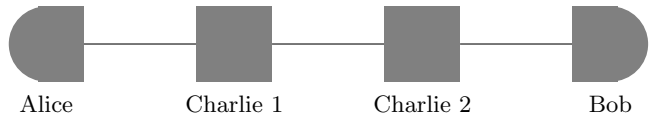

Compared to a standard QKD network with a similar topology, the Qline requires very little hardware. Therefore, the cost per node in the network scales is much better than for point-to-point QKD networks that require each node to be equipped with single photon generators and detectors.

The simplicity of the Qline implies some features that distinguish them from QKD networks. The most important feature is that  Qline can enable key exchange among participants but unlike QKD networks, there is no need for a trusted node set-up assumption. For this reason, the Qline is better suited for local and metropolitan area networks, rather than wide ones. On the other hand, the Qline is fully compatible with QKD networks, and it can easily be integrated into a larger-scale network.

\subsection{Secret sharing and key establishment on Qline}
We introduce the protocol for three parties, but it can be trivially generalized. The operations are described at the qubit level. 

We now describe the Qline protocol for key exchange. As for QKD, the key establishment consists of a quantum communication phase followed by a classical post-processing phase. We present here only the quantum communication and sifting phase, which allows the parties to establish raw keys. In practice, it should be followed with error correction and privacy amplification. Since the raw keys established on Qline are of the same form as the ones established by QKD, the post-processing is completely similar and we don't repeat those steps here.

\subsubsection{Raw key exchange on Qline.}\label{sec:key-exchange-def}

The key exchange protocol can be decomposed into two main steps: The first step leads the parties to output three keys $K_A$, $K_B$ and $K_C$ such that $K_A\oplus K_B \oplus K_C = 0$. This is the secret-sharing phase.
Then, to establish a pair of shared keys between two parties, it suffices for the third one to disclose its key.

We consider that the protocol is scheduled such that each pair of participants gets, on average, the same amount of keys. This is achieved with a key-establishment policy that consists in choosing the party that reveals its key following a round-robin scheduling. Again, by changing the key length and the scheduling policy, it is trivial to adapt the protocol to a weighted round-robin, allowing some pairs of nodes to get more keys than others. This does not affect the overall security, as the proof, we give here can be easily adapted.
\vspace{1em}\\
\noindent
{\bf Protocol}\\
\textit{Input Alice}: an integer $N$;
  \textit{Charlie}: $N$;
  \textit{Bob}: $N$.\\
  \textit{Output Alice}: a key $K_A$; \textit{Charlie}: a key $K_C$; \textit{Bob}: a key  $K_B$.
  
  \begin{description}
  \item[Initialization/pre-generation]
    \begin{itemize}
    \item Alice chooses uniformly at random $2N$ bits: $b^A=(b^A_1,\cdots b^A_N) \in \{0,1\}^N$ and $r^A=(r^A_1,\cdots r^A_N) \in \{0,1\}^N$.
   \item Charlie chooses uniformly at random $2N$ bits: $b^C=(b^C_1,\cdots b^C_N) \in \{0,1\}^N$ and $r^C=(r_1^C,\cdots r_N^C) \in \{0,1\}^N$
   \item Bob chooses uniformly at random $N$ bits $b^B=(b^B_1,\cdots b^B_N) \in \{0,1\}^N$
    \end{itemize}
  \item[State Distribution] 
    \begin{itemize}
    \item Alice prepares $N$ quantum systems in the $\ket{0}$ state, applies rotations $R_Y(\theta_i^A)$ with $\theta_i^A=b^A_i\pi/2+r^A_i\pi$ for $i=0,\dots,N-1$, and sends $\ket{\phi}^A=\bigotimes_{i=0}^{N-1}R_Y(\theta_i^A)\ket{0}_i$ to Charlie. 
    \item For $i=0,\dots,N-1$, Charlie applies $R_Y(\theta_i^C)$ with $\theta_i^C=b_i^C\pi/2+r_i^C\pi$ on $\ket{\phi}^{A}_i$ and sends $\ket{\phi}^{C}=\bigotimes_{j=0}^{N-1}R_Y(\theta_i^C)\ket{\phi}^{A}_i$ to Bob
    \item For $i=0,\dots,N-1$, Bob applies $R_Y(\theta_i^B)$ with $\theta_i^B=b_i^B\pi/2$ on $\ket{\phi}^{C}_i$ and measures $R_Y(\theta_i^B)\ket{\phi}^{C}_i$ in the computational basis, denoting the measurement result $r_i^B\in\{0,1\}$. 
    \end{itemize}
  \item[Public announcement]
  \begin{itemize}
      \item Through an authenticated public channel,
  Alice, Charlie and Bob share their basis choices
 $(b_j^A)$, $(b^C_j)$ and $(b^B_j)$.
 \item Alice updates the sifted key register $K_A$ with 
 \[K_A=(r^A_i)_{i \text{ s.t. } b_i^A \oplus b_i^C \oplus b_i^B = 0}\]
 \item Bob updates the sifted key register $K_B$ with
 \[K_B=(r^B_i)_{i \text{ s.t. } b_i^B \oplus b_i^C \oplus b_i^B = 0}\]
 \item Charlie updates the sifted key register  $K_C$ with
     \[K_C=\left(r_i^C\oplus \frac{b_i^A + b_i^C + b_i^B \mod 4}{2}\right)_{i \text{ s.t. } b_i^A \oplus b_i^C \oplus b_i^B = 0}\]
  \end{itemize}
  
\end{description}

\vspace{1em}
\noindent
{\bf Correctness of the protocol}\\
We assume here that the final state is received by Bob with no losses and errors. While this is unrealistic, this does not affect the security proof.
For all $i$, the final state before the measurement is
\begin{flalign*} & R_Y\Big( (b_i^A+b_i^C+b_i^B)\frac\pi2  + (r^A_i + r_i^C)\pi\Big)\ket{0}\\
   = & R_Y\Big(\big(\frac{b_i^A+b_i^C+b_i^B}{2} + (r^A_i + r_i^C)\big)\pi\Big)\ket{0}\\
   = & Y^{\frac{b_i^A+b_i^C+b_i^B}{2} \oplus r_i^A \oplus  r_i^C}\ket{0}  & \textrm{if } & b^A_i \oplus b^C_i \oplus b^B_i = 0\\  
   = & i^{r_i^A\oplus r_i^C \oplus \frac{b_i^A+b_i^C+b_i^B}{2}}\ket{r^A_i\oplus r_i^C \oplus \frac{b_i^A+b_i^C+b_i^B}{2}}. & 
\end{flalign*}
By selecting only the outputs $i$ such that 
$b^A_i \oplus b^C_i \oplus b^B_i = 0$, we get that the final
measurement outputs
\[r_i^B = r_i^A \oplus r_i^C \oplus \frac{b_i^A+b_i^C+b_i^B}{2},\]
so that the final keys satisfy
\[K_A \oplus K_B \oplus K_C = 0.\]

As previously mentioned, a shared key can then be established between two parties by publicly disclosing the key of the third party. For simplicity, we consider round-robin scheduling. After repeating the protocol, we can therefore consider that the final state consists of three keys $K_1$, $K_2$ and $K_3$ shared by Alice and Bob, Alice and Charlie, and Charlie and Bob, respectively.

     
       
       
       
       
       
       
       

\section{Security Proof of key establishment on Qline}
\label{sec:proof}
In this section, we provide a formal security proof of key establishment over Qline in the \emph{abstract cryptography} framework.
The key establishment over Qline is very similar to standard QKD. For this reason, we are able to leverage the security of QKD~\cite{tomamichel2017largely} and obtain the security of our key establishment. 

Our proof consists of the following steps: 
In section~\ref{sec:assumptions},
we introduce the assumptions required for proving the security of prepare-and-measure QKD from~\cite{tomamichel2017largely}. These assumptions are also required for Qline in order to prove the security of the key establishment.
In Section~\ref{sec:thm}, we state the main theorem.
We then define the real and ideal functionalities of key establishment over Qline and introduce the simulator in Section~\ref{sec:notations}.
In Section~\ref{sec:lemmas}, we prove the main lemmas for indistinguishability.
Finally, we show that the real key establishment securely realizes its ideal functionality in Section~\ref{sec:indistinguishability}


\subsection{Assumptions} \label{sec:assumptions}

Our security proof mainly relies on the security of prepare-and-measure QKD~\cite{tomamichel2017largely}.
We thus require similar assumptions for the key establishment over Qline to be secure.

\begin{assumption}
\label{assump:general}
The following assumptions are required  both for entangled-based and prepare-and-measure QKD protocols, which we also adopt here:
\begin{itemize}
    \item \textbf{Finite-dimensional quantum systems:} The players' Hilbert space can be represented with finite-dimensional Hilbert spaces.
    \item \textbf{Sealed laboratories:} The players are in spatially separated laboratories. This allows writing the joint system of Alice, Charlie and Bob in a tensor product form.
    \item \textbf{Random seeds:} Each player has access to some uniform random seed.
    \item \textbf{Authenticated communication channel.}
\end{itemize}
\end{assumption}

We now introduce the assumptions specifically required for 
the security of prepare-and-measure QKD.

\begin{assumption}[Alice's preparation]
    \label{assump:Alice}  In each round $i \in [N]$, denote $b \in \{0,1\}$ the basis choice and $r \in \{0,1\}$ the bit value.
    Alice's preparation satisfies the following assumptions.
    \begin{itemize}
        \item The quantum state $\rho^{b,r}_{A_i}$ produced by Alice, doesn't leak any information about the basis choice: 
    \begin{equation}\label{eq:qkd-cond-alice-1}
        \sum_{r\in\{0,1\}} \rho^{0,r}_{A_i} = \sum_{r\in\{0,1\}} \rho^{1,r}_{A_i}.
    \end{equation}
    \item The states of the preparation are sufficiently complementary, that is, there exists a constant $\overline{c} < 1$ such that 
    \begin{equation}\label{eq:qkd-cond-alice-2}
        c_i := c(\{\rho^{0,r}_{A_i}\}_r, \{\rho^{1,r'}_{A_i}\}_{r'}) \leq \overline{c},
    \end{equation}
    where
    \[
        c(\{\rho^r\}_r, \{\sigma^{r'}\}_{r'}) := \max_{r,r'} \parallel \sqrt{\rho^r}X^{-1}\sqrt{\sigma^{r'}} \parallel^2_{\infty},
    \]
    for states satisfying $ \sum_{r}\rho^r = \sum_{r'} \sigma^{r'}$ and $X:= \sum_{r} \rho^r$.
    \end{itemize}

\end{assumption}

The following assumption is for Bob's measurements. While our protocol does not consider inconclusive outcomes, we take them into account for the security proof.
Such inconclusive outcomes occur when a photon is expected, but either none is detected (photon loss), or more than one is detected (dark count). 

\begin{assumption}[Bob's measurement]
    \label{assump:Bob}
    \begin{itemize}
        \item Bob's quantum system can be decomposed as $B = B_1 B_2 \dots B_M$ and Bob's measurement can be represented as operators acting on the individual subsystems. 
     
         \item Bob's measurements satisfy:
    \begin{equation}\label{eq:qkd-cond-bob}
        (M^{0, \emptyset}_{B_i})^{\dagger} (M^{0, \emptyset}_{B_i}) = (M^{1, \emptyset}_{B_i})^{\dagger} (M^{1, \emptyset}_{B_i})
    \end{equation}
    where $\{M^{b,r}_{B_i}\}_{r\in\{0,1,\emptyset\}}$ is the generalized measurement on subsystem $B_i$ with setting $b\in\{0,1\}$, $r$ being the outcome bit and $\emptyset$ corresponding to an inconclusive outcome. 
    \end{itemize}
\end{assumption}
We need one additional assumption that is specific to Qline. This assumption prevents an eavesdropper from injecting large quantum systems into Charlie to probe its rotation.
\begin{assumption}
    \label{assump:qline}
    The system that Charlie receives and sends is composed of qubits. 
\end{assumption}
This assumption can be bypassed in real systems, in which Charlie could in principle apply its rotation to multiple qubits without realizing it. An eavesdropper could use
this at its advantage to learn the rotation performed by Charlie.
This attack is reminiscent to Trojan-horse attacks on standard QKD systems, which exploit the reflectivity of optical components.
We thus consider this as a side-channel attack, which will be treated 
in a separate work.

\subsection{Main theorem} \label{sec:thm}

In~\cite{tomamichel2017largely}, it was proved that a QKD protocol
that satisfies Assumptions~\ref{assump:general},~\ref{assump:Alice}~and ~\ref{assump:Bob} is $\epsilon$-secure in the abstract cryptography framework. 
In~\cite{tomamichel2017largely}, the parameter $\epsilon$
is an explicit function of several parameters of the protocol
such as $\bar c$, and $N$ but also depends on the parameters of the post-processing.
More details about this relationship can be found in~\cite{tomamichel2017largely} Equations (56) to (58).
The following theorem stating the security of Qline considers the same $\epsilon$.

\begin{theorem}\label{th:qline-security}
Under Assumptions~\ref{assump:general},~\ref{assump:Alice},~\ref{assump:Bob}~and~\ref{assump:qline}, the key establishment protocol on Qline is $3\epsilon$-secure in the abstract cryptography framework, where $\epsilon$ is the security parameter of a standard QKD protocol under the same 
assumptions.
\end{theorem}


\subsection{Setup and notations} \label{sec:notations}
    Let $S$ denote the size of all output keys in the systems we describe.
        \subsubsection{Basic systems}
        \begin{itemize}
            \item A one-way quantum channel $\Q$ with one input interface for the sender, one output interface for the receiver, and one outer interface. $\Q$ provides external entities full control over the channel. This is modelled by the fact that any input of the sender interface of $\Q$ is output to the outer interface, and the output of the receiver interface directly comes from the input of the outer interface.
            \item A two-way classical channel $\A$ authenticated by a pre-shared key between the two communicating interfaces. $\A$ provides external entities with the ability to read or block the data, but not to tamper with it. This is modelled by an outer interface providing as output a copy $t$ of each message going through the channel, and receiving a binary input $l$ called a blocking lever which, if pulled, prevents the messages to pass through.
            \item An ideal key expansion system $\K$ with three user interfaces. Each one takes as input a short key ($PSK_1$, $PSK_2$ and $PSK_3$), and two of them return back a longer key output ($K_1$ and $K_2$). If $PSK_1 = PSK_2 = PSK_3$, then the outputs $K_1$ and $K_2$ are equal to truly random bit strings of size~$S$. Otherwise, the systems abort with outputs ``$\bot$''. The idealness of $\K$ is captured by the fact that the external entities only get to know the size of the outputs, and can only perform one action which is to force the system to abort, thus preventing the keys to be generated. Accordingly, the outer interface of $\K$ consists of an output $s$ giving the size of the keys ($S$ or $0$), as well as a blocking lever $l$. If the blocking lever is pulled ($l=1$), then the system aborts regardless of any other input.
        \end{itemize}
        \subsubsection{Main systems}
        \begin{itemize}
            \item The systems 
            {$\pi^{Qline}_A$}, $\pi^{Qline}_C$ and $\pi^{Qline}_B$ realize the tripartite key establishment protocol on Qline described in section \ref{sec:key-exchange-def}. They respectively act as Alice, Charlie and Bob. As such, they each have one \textit{user} interface for inputs and outputs of the protocol, as well as \textit{inner} interfaces to communicate with each other. They have no \textit{outer} interfaces as they are considered safe regarding external entities.
            
            The user interfaces of $\pi^{Qline}_A$, $\pi^{Qline}_C$ and $\pi^{Qline}_B$ consists in a short key input ($PSK_A$, $PSK_C$ and $PSK_B$ respectively) and a long key output ($K_A$, $K_C$ and $K_B$ respectively) of size $S$.
            The inner interfaces consist in the required inputs and outputs to be sequentially composed with $\A$ and $\Q$.
            $\pi^{Qline}_A$ and $\pi^{Qline}_B$ have one inner interface while $\pi^{Qline}_C$ has two.
            
            We then consider 3 sub-protocols: $\QL_{AC}$, $\QL_{CB}$ and $\QL_{AB}$ which each comes with its versions $\pi^{Qline_{AC}}$, $\pi^{Qline_{CB}}$ and $\pi^{Qline_{AB}}$ of the three systems described above.
            \begin{align*}
                \QL_{AC} &:= \pi^{Qline_{AC}}_A(\Q||\A)\pi^{Qline_{AC}}_C(\Q||\A)\pi^{Qline_{AC}}_B \\
                \QL_{CB} &:= \pi^{Qline_{CB}}_A(\Q||\A)\pi^{Qline_{CB}}_C(\Q||\A)\pi^{Qline_{CB}}_B \\
                \QL_{AB} &:= \pi^{Qline_{AB}}_A(\Q||\A)\pi^{Qline_{AB}}_C(\Q||\A)\pi^{Qline_{AB}}_B
            \end{align*}
            In these protocols, one of the players, respectively Bob, Alice and Charlie reveal its key to the others via the authenticated channels between the post-selection step and the post-processing step. One of the other players then adds (bit-wise XOR) the revealed key to its own, and performs the post-processing with the remaining player, so that they output two equal final keys. The player who reveals their secret does not have a key output at its user interface (see figure~\ref{fig:Qline}).
            
            The full key establishment on Qline is then defined by 

            \begin{equation} \label{eq:qline_def}
                \QL_{AC} || \QL_{CB} || \QL_{AB}
            \end{equation}

            \item The systems $\sigma^{Qline_{AC}}$, $\sigma^{Qline_{CB}}$ and $\sigma^{Qline_{AB}}$ are simulators. They have an inner interface meant to be plugged into the outer interface of $\K$, involving the size $s$ of the output key as well as the blocking lever $l$ enabling to make $\K$ abort.
            They also have an outer interface matching the one of $\QL_{AC}$, $\QL_{CB}$ and $\QL_{AB}$ respectively. This interface consists in the following inputs and outputs: $\splitatcommas{\{\rho_0', \rho_1', \rho_2', \rho_3', t_0', t_1', l'_0, l'_1\}}$.
        
            We build the simulators as shown in Figure~\ref{fig:K.sigma}.
            In order to produce inputs and outputs of the outer interface, $\sigma^{Qline}$ simulates the Qline by internally running the key establishment and exposing the outer interface of that Qline as its own outer interface. The variants of the simulator $\sigma^{Qline_{AC}}$, $\sigma^{Qline_{CB}}$ and $\sigma^{Qline_{AB}}$ respectively run the $\QL_{AC}$, $\QL_{CB}$ and $\QL_{AB}$ protocols. Regarding the user interfaces of the simulated Qline, the simulator inputs at each of them equal pre-shared keys $PSK'_A = PSK'_C = PSK'_C$, and ignores the output.
        
            In order to prevent             
            $\K$ from aborting while
             the simulated Qline successfully terminates, or vice-versa, we set the following additional behaviours of $\sigma^{Qline}$:
            \begin{itemize}
                \item If the simulated Qline aborts for any reason, the simulator triggers the blocking lever $l$ of $\K$, thus making it also abort.
                \item In the event where $s=0$, indicating that $\K$ aborted for having been provided different pre-shared key inputs, then $\sigma^{Qline}$ inputs different pre-shared keys $PSK'_A$, $PSK'_C$ and $PSK'_B$ at the user interfaces of the simulated Qline.
            \end{itemize}

        \end{itemize}
        

    \begin{figure}[!ht]
        \centering 
        \includegraphics[scale=0.44]{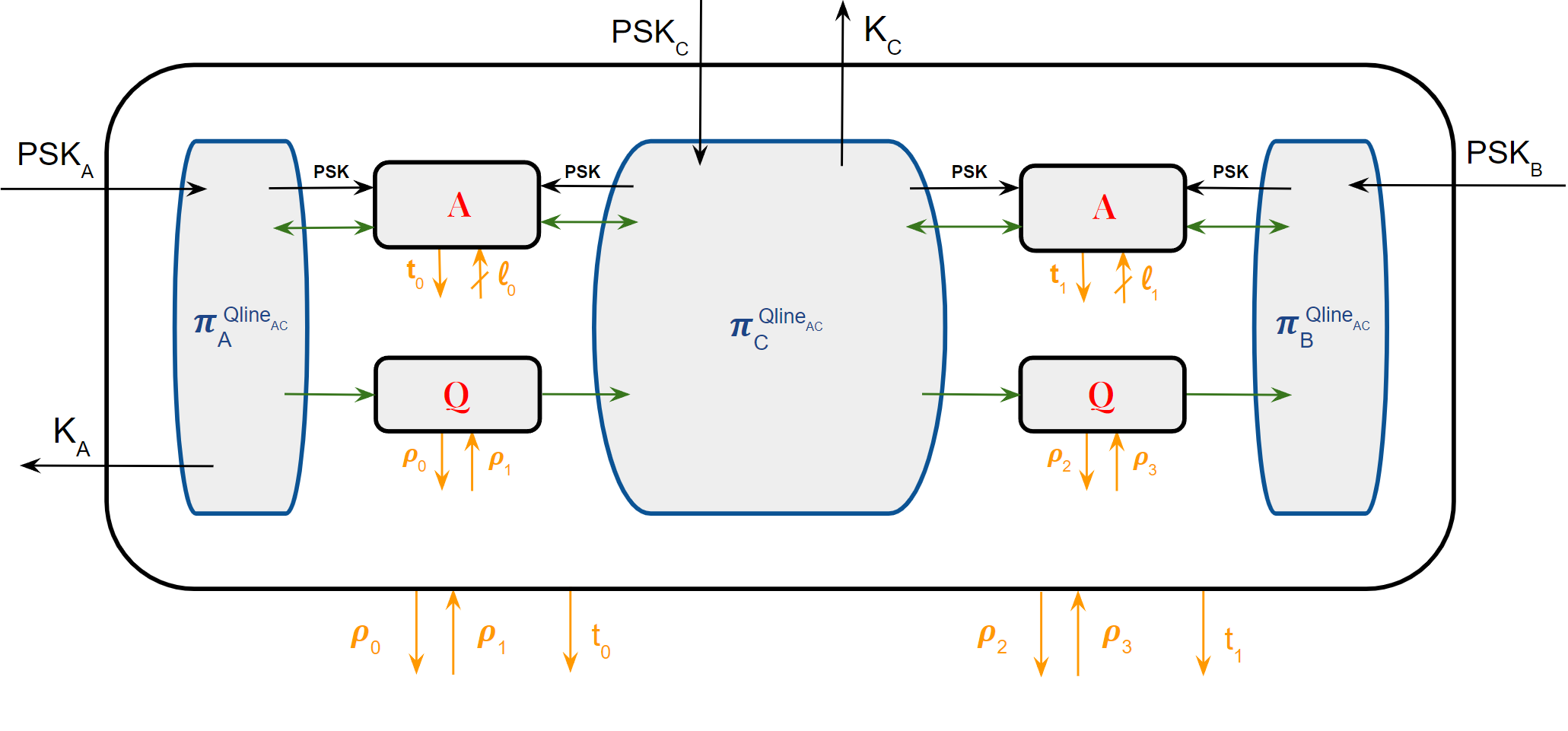}
        \caption{\centering  The $\QL_{AC}$ variant of Qline}
        \label{fig:Qline}
    \end{figure}

    \begin{figure}[!ht]
        \centering \hspace*{-1.25cm}
        \includegraphics[scale=0.33]{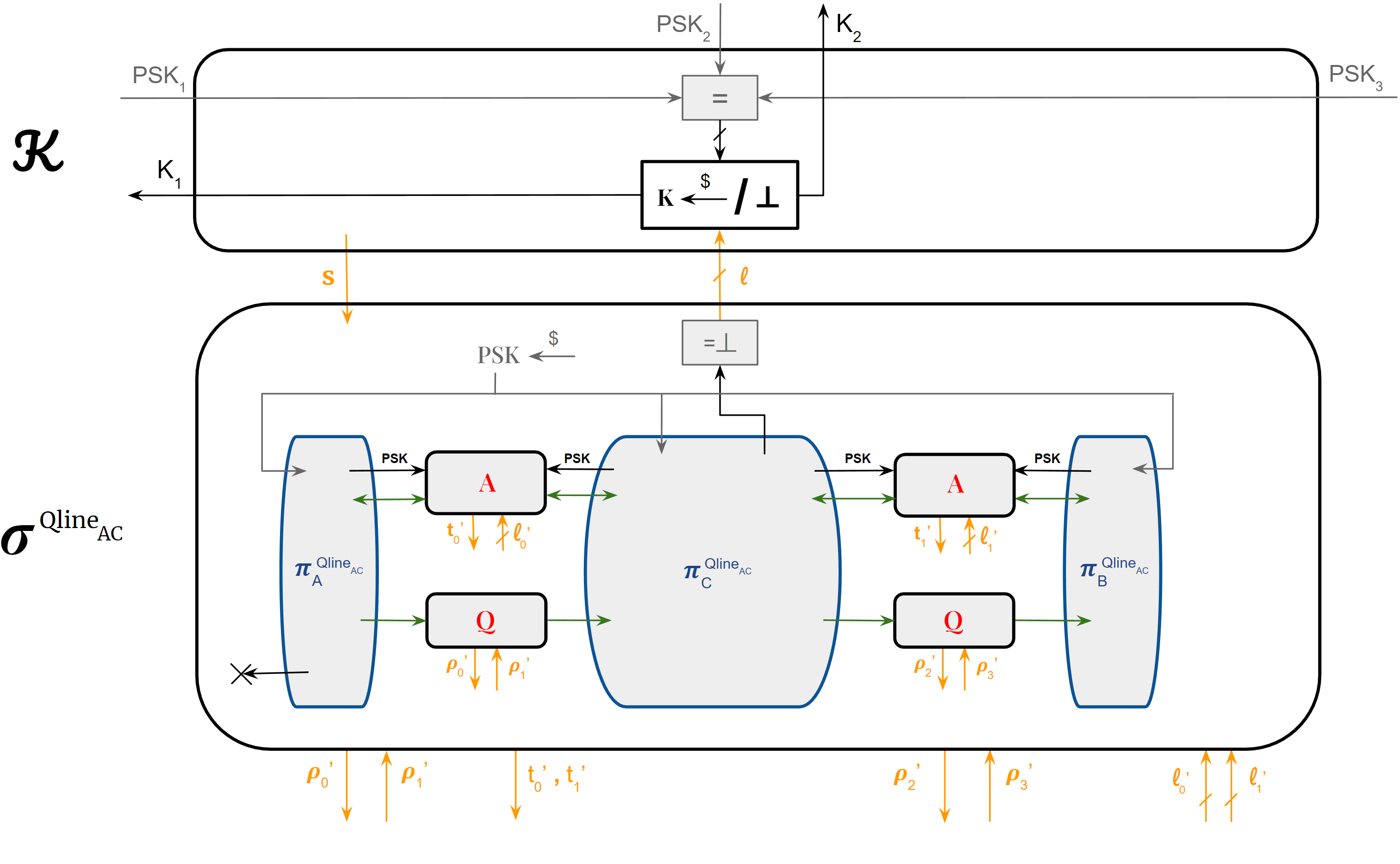}
        \caption{\centering \small The simulator $\sigma^{Qline_{AC}}$ plugged on $\K$}
        \label{fig:K.sigma}
    \end{figure}

\subsection{Preliminary lemmas} \label{sec:lemmas}
    
    Before proving Theorem~\ref{th:qline-security}, we introduce two lemmas.
    These lemmas describe how, in $\QL_{AC}$ (respectively $\QL_{CB}$), after Bob (respectively Alice) has revealed its key, the resulting states and measurements
    satisfies all the assumptions of a secure prepare-and-measure QKD protocol between Alice and Charlie (respectively Charlie and Bob) introduced in Section~\ref{sec:assumptions}.
    
    In order to show that, we prove that the states generated on the Qline are indistinguishable from genuine QKD states. Moreover, we show that the combination of Charlie and Bob in $\QL_{AC}$ and the combination of Alice and Charlie in $\QL_{CB}$ satisfy Assumption~\ref{assump:Bob} on Bob's measurement and
    Assumption~\ref{assump:Alice} on Alice's preparation, respectively.
    
    The first lemma considers the protocol $\QL_{AC}$.
    In this protocol, Charlie does not perform
    any measurements directly. 
    However, we can still consider Charlie's state after the key disclosure stage, and prove that it satisfies the desired properties.
    In other words, we show that the post-measurement states of Charlie, combined with the actual measurements performed by Bob, satisfy the measurement property of Bob in QKD, that is, Assumption~\ref{assump:Bob}.
    
    \begin{lemma}\label{lemma:charlie-is-bob}
    The joint classical-quantum state of $\pi_A^{Qline_{AC}}$ and $\pi_C^{Qline_{AC}}$ is indistinguishable from the joint state of Alice and Bob in prepare-and-measure QKD. Additionally, after $\pi_B^{Qline_{AC}}$ disclosed its key, the combined system $\pi_C^{Qline_{AC}}(\Q||\A)\pi_B^{Qline_{AC}}$ satisfies Assumption~\ref{assump:Bob} on Bob's measurement.
    
    \end{lemma}
    
    \begin{proof}
    
    We first write the post-measurement state of Alice and Bob in QKD~\cite{tomamichel2017largely}:
    \begin{equation}
    \begin{split}
        \tau^{qkd}_{AB} = \frac{1}{8^N} \sum_{r_A, b_A, b_B \in \{0,1\}^N} \sum_{r_B \in \{0,1\}^N} & \ket{r_A,r_B,b_A,b_B}\bra{r_A,r_B,b_A,b_B} \\
        & \otimes \ket{\omega_A(r_B)}\bra{\omega_A(r_B)}\\
        & \otimes M^{b_B, r_B}_B \rho^{b_A,r_A}_B (M^{b_B, r_B}_B)^{\dagger}
    \end{split}
    \label{eq:ABQKD}
    \end{equation}

    
    where $b_A$ and $b_B$ represent the random basis choices of Alice and Bob, $r_A$ is the bit choice of Alice, $r_B$ is the measurement outcome and $\rho^{b_A,r_A}_B$ is the state received by Bob. $\omega_A(r_B)$ is the set of indices $\{i\in [N] \quad|\quad r_B \neq \emptyset\}$ sent by Bob to Alice describing which outcomes were conclusive. 
    
    In order to have the post-measurement state of Charlie in $\QL_{AC}$ after the key announcement stage, we start with the joint state of Alice, Bob and Charlie, after the measurement:
    \begin{equation}
    \small
    \begin{split}
        \tau^{ql}_{ACB} = \frac{1}{32^N} \sum_{b^A, r^A, b^C, r^C, b^B \in \{0,1\}^N} \sum_{r^B \in \{0,1,\emptyset\}^N} & 
        \ket{b^A, r^A, b^C, r^C, b^B, r^B}
        \bra{b^A, r^A, b^C, r^C, b^B, r^B} \\
        & \otimes 
        \ket{\omega_A(r^B), \omega^C}
        \bra{\omega_A(r^B), \omega^C}
        \\
        & \otimes [M^{b^B, r^B} {\rho}_3 (M^{b^B, r^B})^{\dagger}]_B
    \end{split}
    \end{equation}
     Where $\omega_A(r^B)$, as for QKD, is the set $\{i\in [N] \quad|\quad r_B \neq \emptyset\}$ held by Alice, while $\omega^C$ is the complete description of the measurement outcomes that Bob sent to Charlie.
    The state ${\rho}_3$ can be written as follows:
    \begin{equation}
        {\rho}_3 = \bigotimes^N_{i=1}\ket{\phi'_i}\bra{\phi'_i}, \quad \quad \ket{\phi'_i} = R_Y(\theta^C_i)R_Y(\theta^A_i)\ket{0} = R_Y(\theta'_i)\ket{0}
    \end{equation}
    We now trace out Bob's system in Qline, and obtain the following joint state for Alice and Charlie:
    \begin{equation}
    \small
    \begin{split}
        \tau^{ql}_{AC} = tr_B[\tau^{ql}_{ACB}] = \frac{1}{8^N} \sum_{b^A, r^A, b^C \in \{0,1\}^N} \sum_{r^C \in \{0,1\}^N} & \ket{b^A, r^A, b^C, r^C}\bra{b^A, r^A, b^C, r^C} \\
        & \otimes 
        \ket{\omega_A(r^B), \omega^C}
        \bra{\omega_A(r^B), \omega^C}
    \end{split}
    \end{equation}
    We can rewrite the subset of outcomes $\omega^C$ as a density matrix with the POVM operators $M'$ applied to it, as follows:
    \begin{equation}\label{eq:qline-ac-pm}
    \small
    \begin{split}
        \tau^{ql}_{AC} = \frac{1}{8^N} \sum_{b^A, r^A, b^C \in \{0,1\}^N} \sum_{r^C \in \{0,1\}^N} & \ket{b^A, r^A, b^C, r^C}\bra{b^A, r^A, b^C, r^C} \\
        & \otimes 
        \ket{\omega_A(r^B)}
        \bra{\omega_A(r^B)}
        \\
        & \otimes [M'^{b^B, r^B} \rho_3 (M'^{b^B, r^B})^{\dagger}]_C
    \end{split}
    \end{equation}
    This is equivalent to the case where Charlie would have done the measurements locally using Bob's parameters. 
    Since $\theta^A_i$ and $\theta^C_i$ are independent and uniformly random in $\{0, \frac\pi2, \pi, \frac{3\pi}2\}$, $\theta'_i=\theta^A_i+\theta^C_i$ is also uniformly random in $\{0, \frac\pi2, \pi, \frac{3\pi}2\}$. We  thus conclude that $\rho_3$ is 
    a QKD state. Therefore, up to considering that $r^C = \emptyset$ for the inconclusive outcomes, the state of Equation~\ref{eq:qline-ac-pm} is indistinguishable from the
    the state of Equation~\ref{eq:ABQKD}.
    
    
    Additionally, since Assumption~\ref{assump:Bob}
    is satisfied by $\pi_B^{Qline_{AC}}$, $\pi_C^{Qline_{AC}}(\Q||\A)\pi_B^{Qline_{AC}}$ also satisfies it when considering the state after the public announcement phase. This completes the proof. 
    
    \qed
    \end{proof}

Lemma~\ref{lemma:charlie-is-alice} shows the indistinguishability of the states for the run between Charlie and Bob, where we need to show that Alice and Charlie can be grouped together to perform similar operation as an Alice in QKD.

\begin{lemma}\label{lemma:charlie-is-alice}
The classical-quantum state of $\pi_C^{Qline_{CB}}$ is indistinguishable from the state of Alice in prepare-and-measure QKD. Additionally, after $\pi_A^{Qline_{CB}}$ disclosed its key, the combined system $\pi_A^{Qline_{CB}}(\Q||\A)\pi_C^{Qline_{CB}}$ satisfies  Assumption~\ref{assump:Alice} on Alice's preparation.
\end{lemma}

\begin{proof}
First we show that the state output by Charlie satisfies Assumption~\ref{assump:Alice} on Alice's preparation. Let the state $\rho_i$ be the state that Charlie receives the $i$-th round of the protocol. Also the bit $b^C_i = \{0,1\}$ and $r^C_i = \{0,1\}$ represent Charlie's choice of rotation. We denote the rotated state $\rho^{b,r}_{C_i}$, which can be one of the four following ones:
\begin{equation}
\begin{split}
    & \rho^{0,0}_{C_i} = \rho_i \\
    & \rho^{0,1}_{C_i} = R_Y(\pi) \rho_i R_Y(\pi)^{\dagger} = Y \rho_i Y \\
    & \rho^{1,0}_{C_i} = R_Y(\pi/2) \rho_i R_Y(\pi/2)^{\dagger} = Y^{1/2} \rho_i Y^{3/2} \\
    & \rho^{1,1}_{C_i} = R_Y(3\pi/2) \rho_i R_Y(3\pi/2)^{\dagger} = Y^{3/2} \rho_i Y^{1/2}
\end{split}
\end{equation}

Computing the sum over basis choice, we get
\begin{equation} \label{eq:BasisChoiceSum0}
    \sum_{r\in\{0,1\}}\rho^{0,r}_{C_i} = \rho^{0,0}_{C_i} + \rho^{0,1}_{C_i}  = \rho_i + Y \rho_i Y
\end{equation}
and 
\begin{equation} \label{eq:BasisChoiceSum1}
    \sum_{r\in\{0,1\}}\rho^{1,x}_{C_i} = \rho^{1,0}_{C_i} + \rho^{1,1}_{C_i}  = Y^{1/2} \rho_i Y^{3/2} + Y^{3/2} \rho_i Y^{1/2} = \rho_i + Y \rho_i Y.
\end{equation}
Equation~\ref{eq:qkd-cond-alice-1} of the assumption is thus satisfied. 

We now check Equation~\ref{eq:qkd-cond-alice-2}:
\begin{equation} \label{eq:alice's-cond-bob}
     c'_i = \max_{r,r'} \parallel \sqrt{\rho^{0,r}_{C_i}}(\rho_i + Y \rho_i Y)^{-1}\sqrt{\rho^{1,r'}_{C_i}} \parallel^2_{\infty} < \bar{c}.
\end{equation}

Notice that Charlie's set of rotations applied to $\rho_i$ is the same as Alice's set of rotations applied to $\ket 0 \bra 0$. 
Then $(\{\rho^{0, r}_{C_i}\}_{r}, \{\rho^{1, r'}_{C_i}\}_{r'}) = (\{\rho^{0, r}_{A_i}\}_{r}, \{\rho^{1, r'}_{A_i}\}_{r'})$ and thus $c_i' = c_i$,
where $c_i$ is defined in Assumption~\ref{assump:Alice}.
As Equation~\ref{eq:qkd-cond-alice-2} is verified by assumption, Equation~\ref{eq:alice's-cond-bob} is satisfied.

$\pi_A^{Qline_{CB}}(\Q||\A)\pi_C^{Qline_{CB}}$ thus satisfies Assumption~\ref{assump:Alice} on Alice's preparation.\\



We now show that the state of $\pi_C^{Qline_{CB}}$ is indistinguishable from the one of Alice in prepare-and-measure QKD. The full state of Alice's system in QKD before the distribution phase, including her classical registers, is:
\begin{equation}\label{eq:alice-full-state-qkd}
    \tau^{qkd}_{A} = \frac{1}{4^N} \sum_{b_A,r_A \in \{0,1\}^N} \ket{b_A}\bra{b_A} \otimes \ket{r_A}\bra{r_A} \otimes \rho^{b_A,r_A}_A
\end{equation}
where $\rho^{b_A, r_A}_A = \bigotimes_{i=1}^N \rho^{b_{Ai}, r_{Ai}}_A$ is the state that Alice sends to Bob.

In Qline, the full state of Charlie before he sends his state is: 
\begin{equation}\label{eq:charlie-full-state-qline}
    \tau^{ql}_{C} = tr_A[\tau^{ql}_{AC}] = \frac{1}{4^N} \sum_{b^C, r^C \in \{0,1\}^N} \ket{b^C}\bra{b^C} \otimes \ket{r^C}\bra{r^C} \otimes \rho_2
\end{equation}
Where
\begin{equation}
    \rho_2 = \bigotimes^N_{i=1}\ket{\phi'_i}\bra{\phi'_i}, \quad \quad \ket{\phi'_i} = R_Y(\theta^C_i)R_Y(\theta^A_i)\ket{0} = R_Y(\theta'_i)\ket{0}
\end{equation}

Since $\theta^A_i$ and $\theta^C_i$ are independent uniformly random in $\{0, \frac\pi2, \pi, \frac{3\pi}2\}$, $\theta'_i=\theta^A_i+\theta^C_i$ is also uniformly random in $\{0, \frac\pi2, \pi, \frac{3\pi}2\}$. Therefore, we conclude that $\rho_2$ is 
a QKD state, and thus that the classical quantum state $\tau^{ql}_{C}$ of Charlie is perfectly indistinguishable from $\tau^{qkd}_{A}$ the one of Alice in prepare-and-measure QKD.

\qed
\end{proof}
    
\subsection{Proof of Theorem~\ref{th:qline-security}} 
\label{sec:indistinguishability}

Recall that the key establishment on Qline is defined as the parallel composition of $\QL_{AC}$, $\QL_{CB}$ and $\QL_{AB}$ (\ref{eq:qline_def}).  We separately show the three following indistinguishability relations:
\begin{eqnarray}
    \label{eq:MainAC}
    \QL_{AC} 
    &\approx_{\epsilon} 
    &\K\circ\sigma^{Qline_{AC}} \\
    \label{eq:MainCB}
    \QL_{CB} 
    &\approx_{\epsilon} 
    &\K\circ\sigma^{Qline_{CB}} \\
    \label{eq:MainAB}
    \QL_{AB} 
    &\approx_{\epsilon} 
    &\K\circ\sigma^{Qline_{AB}}
\end{eqnarray}

Since in Abstract Cryptography, security is composable~\cite{AbstractCrypto}, this immediately implies that
the Qline protocol is $3\epsilon$-secure.
    
We present the proof of Equation~\ref{eq:MainAC} in full detail. 
We then explain how this can be adapted to establish Equations~\ref{eq:MainCB} and~\ref{eq:MainAB}.

In order to prove Equation~\ref{eq:MainAC}, we consider a distinguisher $\D$ which is given a system $\S$, either equal to $\QL_{AC}$ or to $\K\circ\sigma^{Qline_{AC}}$, uniformly at random.
We denote the input and outputs of $\S$ by $\splitatcommas{\{\tilde{PSK}_A, \tilde{PSK}_C, \tilde{PSK}_B, \tilde{K}_A, \tilde{K}_C, \tilde{\rho}_0, \tilde{\rho}_1, \tilde{\rho}_2, \tilde{\rho}_3, \tilde{t}_0, \tilde{t}_1, \tilde{l}_0, \tilde{l}_1}\}$.
We split the inputs and outputs into two groups, namely the main key outputs $G = \{\splitatcommas{\tilde{K}_A, \tilde{K}_C}\}$, and the other inputs and outputs $H = \splitatcommas{\{\tilde{PSK}_A, \tilde{PSK}_C, \tilde{PSK}_B, \tilde{\rho}_0, \tilde{\rho}_1, \tilde{\rho}_2, \tilde{\rho}_3, \tilde{t}_0, \tilde{t}_1, \tilde{l}_0, \tilde{l}_1}\}$. 

In the following, we show that $G\cup H$ is independent of whether $\S$ is $\QL_{AC}$ or $\K\circ\sigma^{Qline_{AC}}$. To do so, we separately study the independence of each subset $G$ and~$H$. 

\begin{lemma} \label{lem:H&S-indep}
The set of inputs and outputs $H = \splitatcommas{\{\tilde{PSK}_A, \tilde{PSK}_C, \tilde{PSK}_B, \tilde{\rho}_0, \tilde{\rho}_1, \tilde{\rho}_2, \tilde{\rho}_3, \tilde{t}_0, \tilde{t}_1, \tilde{l}_0, \tilde{l}_1}\}$ is independent of whether $\S$ is $\QL_{AC}$ or $\K\circ\sigma^{Qline_{AC}}$.
\end{lemma}

\begin{proof}

Whatever $\S$ is, the outputs $\{\splitatcommas{\tilde{\rho}_0, \tilde{\rho}_2, \tilde{t}_0, \tilde{t}_1}\}$ come from the same system (i.e. $\QL_{AC}$) which is given the following inputs: 
\begin{itemize}
    \item $\{\splitatcommas{\tilde{PSK}_A, \tilde{PSK}_C, \tilde{PSK}_B, \tilde{\rho}_1, \tilde{\rho}_3, \tilde{l}_0, \tilde{l}_1}\}$ if $\S$ is $\QL_{AC}$
    \strut{}
    \item $\{\splitatcommas{PSK'_1, PSK'_2, PSK'_3, \tilde{\rho}_1, \tilde{\rho}_3, \tilde{l}_0, \tilde{l}_1}\}$ if $\S$ is $\K \circ \sigma^{Qline_{AC}}$
\end{itemize}

\noindent The construction of the simulator guarantees that
$$
\tilde{PSK}_A = \tilde{PSK}_C = \tilde{PSK}_B \Longleftrightarrow PSK'_A = PSK'_C = PSK'_B 
$$
Furthermore, the actual value of the pre-shared key inputs does not influence the protocol in any way and is thus fully independent of the outputs. The only relation that does have an influence on the protocol is their equality. If $PSK_A = PSK_B = PSK_C$, then $PSK'_A = PSK'_B = PSK'_C$ and both protocol successfully proceed. Otherwise, both systems abort, returning identical outputs.

Finally, regardless of what $\S$ is, the outputs $\{\splitatcommas{\tilde{\rho}_0, \tilde{\rho}_2, \tilde{t}_0, \tilde{t}_1}\}$ are those of the same system, which is given inputs that are either identical or that identically influence the system. 
To conclude, $H$ is fully independent of what $\S$ is. \qed
    
\end{proof}

In particular, lemma~\ref{lem:H&S-indep} implies that the conditions for $\S$ to abort are independent of whether $\S$ is $\QL_{AC}$ or $\K\circ\sigma^{Qline_{AC}}$. In such conditions, the main key outputs are set to ``$\bot$'' and the distinguisher $\D$ thus holds no information at all on what $\S$ is.
In the following, we consider that $\S$ successfully terminates, producing actual keys at the main outputs $\tilde{K}_A$ and $\tilde{K}_C$

\begin{lemma} \label{lem:G-indep}
From the point of view of the distinguisher, except with probability $\epsilon$ for some $\epsilon>0$, the set of inputs and outputs $G = \splitatcommas{\{\tilde{K}_A, \tilde{K}_C}\}$ is independent of both the set $H$, and of whether $\S$ is $\QL_{AC}$ or $\K\circ\sigma^{Qline_{AC}}$
\end{lemma}

\begin{proof}
    If $\S = \K\circ\sigma^{Qline_{AC}}$, then the keys $\tilde{K}_A$ and $ \tilde{K}_C$ come from $\K$ and are thus uniformly random. We show that the same applies if $\S = \QL_{AC}$. More precisely, we show that from the point of view of the distinguisher, except with probability $\epsilon$ for some $\epsilon>0$, $G$ appears uniformly random and fully independent of $\S$ and of the other inputs and outputs $H$.
    Suppose that 
    \[\S = \QL_{AC} = \pi_A^{Qline_{AC}}(\Q||\A)\pi_C^{Qline_{AC}}(\Q||\A)\pi_B^{Qline_{AC}}\]
    Lemma~\ref{lemma:charlie-is-bob} guarantees that the states produced by $\S$ between
    $\pi_A^{Qline_{AC}}$ and $\pi_C^{Qline_{AC}}$
    are indistinguishable from QKD states, and that $\pi_C^{Qline_{AC}}(\Q||\A)\pi_B^{Qline_{AC}}$.  satisfies Assumption~\ref{assump:Bob}. By application of~\cite{tomamichel2017largely}, the key establishment protocol is $\epsilon$-secure.
    
    The definition of $\epsilon$-security from~\cite{tomamichel2017largely} implies that, from the point of view of the distinguisher, except with probability $\epsilon$, the final keys $\tilde{K}_A$ and $\tilde{K}_C$ after post-processing are equal uniformly random bit strings. In particular they are completely independent of the total classical-quantum state held by the distinguisher, which includes $H$ the inputs and outputs of $\S$. This concludes the proof.\qed
\end{proof}

We are now ready to prove Equation~\ref{eq:MainAC}.
We want to show that the set $G\cup H$ of inputs and outputs of $\S$ is independent of whether $\S$ is $\QL_{AC}$ or $\K\circ\sigma^{Qline_{AC}}$.
Let $g$ and $h$ be some values of $G$ and $H$ respectively. We have
\begin{equation} \label{eq:probas1}
\begin{split}
Pr[G=g, H=h, \S=\QL_{AC}] & = Pr[G=g | (H=h, \S=\QL_{AC}) ]\\
& \times Pr[H=h, \S=\QL_{AC}]
\end{split}
\end{equation}
Lemma~\ref{lem:H&S-indep} implies that 
\[
Pr[H=h, \S=\QL_{AC}] = Pr[H=h] \times Pr[\S=\QL_{AC}]
\]
Lemma~\ref{lem:G-indep} implies that, from the point of view of the distinguisher, except with probability $\epsilon$, 
\[
Pr[G=g | (H=h, \S=\QL_{AC}) ] = Pr[G=g]
\]
Equation~\ref{eq:probas1} then becomes
\begin{equation} \label{eq:probas2}
Pr[G=g, H=h, \S=\QL_{AC}] =_\epsilon Pr[G=g] \times Pr[H=h] \times Pr[\S=\QL_{AC}]
\end{equation}
where $=_\epsilon$ means that the equality holds except with probability $\epsilon$.
Finally, applying lemma~\ref{lem:G-indep} to Equation~\ref{eq:probas2}, we get that from the point of view of the distinguisher
\begin{equation} \label{eq:probas3}
Pr[G=g, H=h, \S=\QL_{AC}] =_\epsilon Pr[G=g, H=h] \times Pr[\S=\QL_{AC}]
\end{equation}

Because Equation~\ref{eq:probas3} holds for any value $g$ and $h$, this shows that the set $G\cup H$ of inputs and outputs of $\S$ is independent of what $\S$ is.
This gives the desired Equation~\ref{eq:MainAC}.

The proof of Equation~\ref{eq:MainCB}, is similar to the one presented above with the only differences that: 1) $\S$ is instead either $\QL_{CB}$ or $\K\circ\sigma^{Qline_{CB}}$, and 2)
Lemma~\ref{lemma:charlie-is-bob} should be replaced with Lemma~\ref{lemma:charlie-is-alice}. The proof is then similar, and the conclusion is Equation~\ref{eq:MainCB}.

Equation~\ref{eq:MainAB} is more straightforward than the two others as by definition,  $\pi_A^{Qline_{AB}}$ and $\pi_B^{Qline_{AB}}$ respectively satisfy
Assumption~\ref{assump:Alice} about Alice's preparation
and Assumption~\ref{assump:Bob} about Bob's measurement. 
As for equations~\ref{eq:MainAC} and~\ref{eq:MainCB},
the proof can be obtained by considering a system $\S$ that is either $\QL_{AB}$ or $\K\circ\sigma^{Qline_{AB}}$. The conclusion is Equation~\ref{eq:MainAB}.

As a conclusion, Equations~\ref{eq:MainAC},~\ref{eq:MainCB} and~\ref{eq:MainAB} respectively show that $\QL_{AC}$, $\QL_{CB}$ and $\QL_{AB}$ are each $\epsilon$-secure, in the sense that they securely realize the ideal functionality defined by $\K$.
Moreover, the composition theorem of~\cite{AbstractCrypto} implies that the composition of an $\epsilon$-secure and an $\epsilon'$-secure system is $(\epsilon + \epsilon')$-secure.
As the key establishment protocol on Qline amounts to a  parallel composition of $\QL_{AC}$, $\QL_{CB}$ and $\QL_{AB}$, it is $3\epsilon$-secure.
This concludes the proof of Theorem~\ref{th:qline-security}. \qed 
\vspace{1em}

Before concluding, we can hint at how the work presented here scales to more than one Charlie node between Alice and Bob.
As mentioned earlier, the extension of the protocol is trivial. 
Assuming that there are $\ell$ parties in total,
Each intermediate node applies the same protocol as Charlie in the protocol presented here. At the end, each party has a key $K_i$, $i \in [1,\ldots, \ell]$, such that $\bigoplus_i K_i = 0$. It then suffices that all parties except two reveal their keys in order for the two remaining ones to establish a shared key.

Regarding the proof, we also need to consider the additional case in which two intermediate Charlies establish a shared key. The structure of the proof is then similar, with the exception that the security of the state obtained after the key reveal is proved using both Lemma~\ref{lemma:charlie-is-bob} and Lemma~\ref{lemma:charlie-is-alice}. The rest of the proof follows easily.

\section{Conclusion}

We have analysed a key establishment protocol on Qline and proved its security
in the abstract cryptography framework. The closeness between QKD and Qline allowed us to largely reuse the proof technique from~\cite{tomamichel2017largely} for each subprotocols, and combine them
using the composition theorem from~\cite{AbstractCrypto}.

This proof technique implies that, in addition to the security, we get the composability of the security of Qline. In particular, it shows that the key establishment over Qline can be combined with classical authentication and the one-time pad to design a complete secure cryptosystem. It can also be connected to standard QKD architecture to build large-scale quantum network architectures. Finally, it can be combined with classical cryptosystems, as exploited by the Muckle protocol~\cite{DHP20}.

The key establishment protocol derives from a well-known quantum-secret sharing protocol.
Interestingly, the secret sharing protocol was first introduced using multiparty entangled states.
It has then been shown that the correlation obtained after measuring this state could be reproduced by sending a single-qubit state to the parties on a line. Furthermore, a similar architecture has been shown to allow other secure multiparty computing tasks. This seems to indicate that the Qline architecture could be well-suited to execute more protocols than those allowed on QKD networks.

In the shorter term, it seems reasonable to interconnect Qline with large-scale quantum networks.
Compared to standard QKD networks, Qline offers a different tradeoff between security and efficiency. In QKD, the distance between two nodes is limited to approximately 100km. In comparison, this limit applies to the whole Qline. Therefore, the purpose of Qline is not to increase the distance but rather to increase the connectivity of the network.

Another interesting advantage of Qline is the absence of trusted nodes. Trusted nodes are a standard technique to increase the distance of key exchange in QKD networks. These intermediate nodes are used to route the keys to distant parties, and thus get a clear view of the keys they are routing.
This makes these nodes a weak point of the network which needs to be protected with physical security.

In critical infrastructures, such as telecom hubs, most nodes are already highly secured, which makes
them well-trusted nodes. However, when considering end-users of key exchange, imposing physical security on intermediate nodes can be problematic. This makes Qline a good solution to implement the last mile of quantum networks.
By formally proving the composable security of key exchange over Qline, we expand the range of potential end-users of quantum key exchange technologies.

For our proof, we've only considered the simple mathematical abstraction in which parties can only send qubits. It is well known that real QKD systems using attenuated lasers usually consider higher dimension systems, which opens the way the side-channel attacks. In our case, the Charlies can not distinguish between applying their transforms to one or multiple photons. This allows
eavesdropper to perform a Trojan-horse attack in which they inject photons into Charlie and recover them deterministically to learn the transform that was applied.

Many techniques have been developed to protect QKD from side-channel attacks. In particular, differential phase shift or decoy state protocols precisely aim at preventing the use of this discrepancy between the mathematical model of QKD and its physical implementation.
In a separate work, we will show how these can be applied to Qline to protect it against side-channel attacks. More information
about protection against side-channels on Charlie can be found in VeriQloud's patent~\cite{VQpatent}.

\section{Acknowledgements}
All authors acknowledge extensive discussion with Joshua Nunn, Scientific Advisor of VeriQloud.

\bibliographystyle{ieeetr}
\bibliography{MyRef.bib}

\end{document}